\documentclass[11pt]{article}

\pdfoutput=1  

\usepackage{amsmath,amsthm,amsfonts,amssymb,mathrsfs}

\usepackage{geometry}
\usepackage[skip=2pt,font=small]{caption}

\usepackage[usenames,dvipsnames]{color}
\usepackage{tikz,pgfplots,xcolor}

\usepackage{cite}

\usepackage{color}
\usepackage{graphicx}

\usepackage{soul}

\usepackage[colorlinks,%
linkcolor=BrickRed,%
filecolor =BrickRed,%
citecolor=RoyalPurple,%
]{hyperref}

\makeatletter
\newcommand{\specialcell}[1]{\ifmeasuring@#1\else\omit$\displaystyle#1$\ignorespaces\fi}
\makeatother

\def\transpose{{\hbox{\tiny\it T}}}

\renewcommand{\Re}{{\mathbb R}}

\newcommand{\argmin}{\mathop{\rm arg\, min}}

\def\be{\begin{eqnarray}}
\def\ee{\end{eqnarray}}
\def\ben{\begin{eqnarray*}}
\def\een{\end{eqnarray*}}

\def\barc{\oo c}

\def\ddt{{\mathchoice{\FRAC{1}{d}{dt}}%
{\FRAC{1}{d}{dt}}%
{\FRAC{3}{d}{dt}}%
{\FRAC{3}{d}{dt}}}}

\def\ddxp{{\mathchoice{\FRAC{1}{d^+}{dx}}%
{\FRAC{1}{d^+}{dx}}%
{\FRAC{3}{d^+}{dx}}%
{\FRAC{3}{d^+}{dx}}}}

\graphicspath{{TDLearningPaperFigures/}}
\def\Ebox#1#2{%
	\begin{center}
		\includegraphics[width= #1\hsize]{#2} 
	\end{center}}

\def\sq{$\Box$}

\def\qed{\ifmmode\Box\else{\unskip\nobreak\hfil
\penalty50\hskip1em\null\nobreak\hfil\sq
\parfillskip=0pt\finalhyphendemerits=0\endgraf}\fi\par\medbreak}

\newcommand{\bsq}{\rule{1.25ex}{1.25ex}}

\def\bqed{\ifmmode\bsq\else{\unskip\nobreak\hfil
\penalty50\hskip1em\null\nobreak\hfil\bsq
\parfillskip=0pt\finalhyphendemerits=0\endgraf}\fi\medskip}

\newsavebox{\junk}
\savebox{\junk}[1.6mm]{\hbox{$|\!|\!|$}}

\def\state{{\sf X}}

\newcommand{\field}[1]{\mathbb{#1}}

\def\ZZ{\field{Z}}

\def\nat{\field{Z}_+}

\newcommand{\one}{\hbox{\rm\large\textbf{1}}}

\def\bfmN{{\mbox{\protect\boldmath$N$}}} 

	\def\bfmU{{\mbox{\protect\boldmath$U$}}}

\def\bfvarphi{\mbox{\boldmath$\varphi$}}

\def\til={{\widetilde =}}

\def\tiltheta{{\tilde \theta}}

\def\clB{{\cal B}}

\def\clD{{\cal D}}
\def\clE{{\cal E}}

\def\clN{{\cal N}}

\def\clS{{\cal S}}

\def\half{{\mathchoice{\textstyle \frac{1}{2}}%
{\frac{1}{2}}%
{\hbox{\tiny $\frac{1}{2}$}}%
{\hbox{\tiny $\frac{1}{2}$}} }}

\def\eqdef{\mathbin{:=}}

\def\Prob{{\sf P}}

\def\Expect{{\sf E}}

\def\epsy{\varepsilon}

\def\varble{\,\cdot\,}

\newtheorem{theorem}{Theorem}[section]

\newtheorem{proposition}[theorem]{Proposition}
\newtheorem{lemma}[theorem]{Lemma}

\def\Lemma#1{Lemma~\ref{#1}}
\def\Prop#1{Prop.~\ref{#1}}

\def\Section#1{Section~\ref{#1}}

\def\Figure#1{Figure~\ref{#1}}

\newcommand{\oo}{\overline}

\def\barh{{\oo {h}}}

\def\Lv{L_\infty^v}

\def\bfmX{{\mbox{\protect\boldmath$X$}}}

\def\FRAC#1#2#3{\genfrac{}{}{}{#1}{#2}{#3}}

\def\half{{\mathchoice{\FRAC{1}{1}{2}}%
{\FRAC{1}{1}{2}}%
{\FRAC{3}{1}{2}}%
{\FRAC{3}{1}{2}}}}

\def\tilc{\tilde c}

\def\tilpsi{\tilde \psi}

\def\grad{\nabla}
\def\gradpsi{\nabla \psi}

\def\tilgrad{\widetilde \nabla}

\def\bfmath#1{{\mathchoice{\mbox{\boldmath$#1$}}%
{\mbox{\boldmath$#1$}}%
{\mbox{\boldmath$\scriptstyle#1$}}%
{\mbox{\boldmath$\scriptscriptstyle#1$}}}}

\def\bfmB{\bfmath{B}}

\def\bfmN{{\mbox{\protect\boldmath$N$}}}

\def\transpose{{\hbox{\tiny\it T}}}

\newcounter{rmnum}

\newenvironment{arabnum}{\begin{list}{{\upshape \arabic{rmnum}.\ }}{\usecounter{rmnum}
\setlength{\leftmargin}{8pt}
\setlength{\rightmargin}{10pt}
\setlength{\itemindent}{-1pt}
}}{\end{list}}

\newcounter{anum}
\newenvironment{alphanum}{\begin{list}{{\upshape (\alph{anum})}}{\usecounter{anum}
\setlength{\leftmargin}{8pt}
\setlength{\rightmargin}{10pt}
\setlength{\itemindent}{-1pt}
}}{\end{list}}

\newenvironment{Anum}[1]{\smallbreak
\noindent
{\bf{Assumption~{#1}}:}
\begin{list}{{\upshape \textbf{#1.\arabic{anum}}:}}{\usecounter{anum}
\setlength{\leftmargin}{12pt}
\setlength{\rightmargin}{8pt}
\setlength{\itemindent}{15pt}
}}{\qed\end{list}}

\newlength{\noteWidth}
\long\def\notes#1{\ifinner
{\tiny #1}
\else
\marginpar{\parbox[t]{\noteWidth}{\raggedright\tiny #1}}
\fi}

\def\notes#1{}

\def\spm#1{\notes{spm:  #1}}

\def\generate{{\cal D}}

\def\DeltaA{\Delta}  

\def\gh{\grad h}
\def\fee{\text{f}}

\def\dffA{{\cal A}}
\def\Sens{{\cal S}}

\begin{document}

\title{Differential TD Learning 
for Value Function Approximation
}

\author{Adithya M.\ Devraj and Sean P.\ Meyn
\thanks{A.D. and S.M. are with the Department of Electrical and Computer
Engg.\ at the University of Florida, Gainesville.
Research is supported by the NSF grants CPS-0931416 and CPS-1259040.
Extended version of paper
submitted to IEEE Conference on Decision \&\ Control, March, 2016.}
}

\maketitle
\thispagestyle{empty}
\setcounter{page}{0}
 
\begin{abstract} 
Value functions arise as a component of algorithms as well as performance metrics in statistics and engineering applications. Computation of the associated Bellman equations is numerically challenging in all but a few special cases.   


A popular approximation technique is known as Temporal Difference (TD) learning. 
The  algorithm introduced in this paper is intended to resolve two well-known problems with this approach:  In the discounted-cost setting, the variance of the algorithm diverges as the discount factor approaches unity.   Second, for the average cost setting, unbiased algorithms exist only in special cases. 

It is shown that the gradient of any of these value functions admits a representation that lends itself to algorithm design.   Based on this result, the new \textit{differential TD} method is obtained for Markovian models on Euclidean space with smooth dynamics. 


Numerical examples show remarkable improvements in performance.  In application to speed scaling, variance is reduced by two orders of magnitude.   

%

\medskip

	\smallskip

{\small

}

\end{abstract}

 
\section{Introduction}
\label{sec:Intro}

The value functions considered in this paper are based on a Markov chain  $\bfmX = \{X(t): t=0,1,2,\ldots\},$ taking values in $\Re^d$, and an associated cost function $c:\Re^d \to \Re$.  A critical modeling assumption is the evolution equation, 
\begin{equation}
X(t+1) = a(X(t),N(t+1)), \quad t \in \nat,
\label{e:SP_disc}
\end{equation}
in which $\bfmN = \{N(t): t=0,1,2,\ldots\}$ is an $m$-dimensional i.i.d.\ disturbance sequence, and $a: \Re^{d\times m} \to \Re^d$ is continuous.    Under these assumptions,  $X(t+1)$ is a continuous function of its initial condition $X(0)$;  this observation is the starting point for the construction of algorithms for value function approximation.

\spm{Notation:   we have settled on $d$ for the dimension of $\state$.  Be warned: I have swapped $\ell$ for $l$ for the dimension of $\theta$}
We begin with some familiar background.

\subsection{Value functions in control and statistics}

A common performance metric in stochastic control and finance is the total discounted cost:
\begin{equation}
h_\alpha(x)=  \sum_{t=0}^{\infty} \alpha^t \Expect [c(X(t)) \mid X(0) = x]\, ,
\label{e:DCOE_disc}
\end{equation}
where $\alpha \in (0,1)$ is the discount factor.
The average cost is defined as the ergodic limit,
\begin{equation}
\barc = \lim_{n\rightarrow\infty} \frac{1}{n} \sum_{t=0}^{n-1} \Expect [c(X(t)) \mid X(0) = x]\, ,
\label{e:barc}
\end{equation}
which is of interest in many areas  beyond control engineering.   The following \emph{relative value function} is central to analysis of the average cost:
\begin{equation}
h(x)=  \sum_{t=0}^{\infty} \Expect [\tilc(X(t)) \mid X(0) = x],
\label{e:fish-sum}
\end{equation}
where $\tilc = c - \barc$. In particular,  under general conditions, the asymptotic variance (the variance appearing in the Central Limit Theorem for the ergodic average \eqref{e:barc}) can be expressed in terms of the relative value function \cite{CTCN}.

Under the assumptions imposed in this paper, the average cost  is deterministic and independent of $X(0)=x$.   Moreover, the relative value function solves \emph{Poisson's equation}: 
\begin{equation} 
 \Expect[h(X(t+1)) - h(X(t)) \mid X(t) = x ] = - \tilc(x)\, .
\label{e:ACOE_disc2}
\end{equation} 

Significant applications include,

\textit{Optimal control:}   
The policy iteration algorithm is used to compute an optimal policy based on two steps.   Given a policy, it is first necessary to obtain the associated value function.  The second step is to update the policy based on this value function \cite{bershr96a}.   This approach is used for both discounted and average-cost optimal control problems.

 Poisson's equation finds application in many other fields:

 \spm{
\cite{bertsi96a}:  TD/Q-learning,   and \cite{bershr96a}:  basic MDP theory
\\
 find a better place for this statement: it so happens in policy iteration that we are only interested in the \emph{{gradients}} of the value functions, rather than the value functions themselves. 
}

\spm{didn't like ranking of states -- I don't think this would be accepted}

\spm{This is getting into detail for something we know little about, and has little to do with the paper.   Can we write a more concise statement about performance metrics in finance?
Estimating the value function also has applications in time-series prediction. For example, the problem of estimating the current value of a business model as a discounted sum of its future cash flows, based on it's current policies is of primary interest in finance.}

\textit{Variance reduction:}  
 The control variate method is intended to reduce variance for various Monte-Carlo methods;
a version of this technique involves the construction of an approximate solution to Poisson's equation \cite{asmgly07,CTCN}.

\spm{cut out for CDC hengly01b, laumehmeyrag15}

\textit{Nonlinear filtering:}  A recent approach to approximate nonlinear filtering requires the solution to Poisson's equation to obtain the innovation gain \cite{yanmehmey13}.  Approximations are required for efficient implementation of this method.


We next recall the basic ideas surrounding TD-learning algorithms for value function approximation. 
The discussion is restricted to discounted-cost value functions.

\subsection{TD-learning and value function approximation}

Closed-form expressions for any of the value functions \eqref{e:DCOE_disc} or \eqref{e:ACOE_disc2} is impossible in all but a few special cases, such as linear systems with quadratic cost. One approach to approximation is a simulation based algorithm known as Temporal Difference (TD) learning  \cite{sutbar98,bertsi96a}.

The goal of TD-learning is to approximate the value function $h_\alpha$ using a parameterized family of functions $\{h_\alpha^\theta: \theta\in\Re^\ell\} $.   Throughout most of the paper we restrict to a linear parameterization of the form 
\begin{equation}
 h_\alpha^\theta = \sum_{j=1}^\ell \theta_j   \psi_j,
\label{e:hLinearPar}
\end{equation}
where $\psi\colon\Re^d\to\Re^\ell$ is continuously differentiable.
 The optimal parameter $\theta^*$ is the solution to a minimum-norm problem,
\begin{equation}
\begin{aligned}
\theta^* &= \argmin_\theta \| h_\alpha^\theta - h_\alpha\|^2_\pi 
\\
& =  \argmin_\theta \Expect[ (h_\alpha^\theta(X) - h_\alpha(X))^2 ],  
\end{aligned}
\label{e:TDgoal}
\end{equation}
where the expectation is with respect to $X\sim \pi$, the steady-state distribution;
see \Section{s:setup} for details.

 \spm{This tutorial will be useful for TD learning \cite{huachemehmeysur11} -- please review!
}

Theory for TD learning in the discounted cost setting is largely complete, in the sense that criteria for convergence are well-understood, and the asymptotic variance of the algorithm is computable based on standard theory from stochastic approximation theory \cite{bor08a}.   Theory and algorithms for the average cost setting is more fragmented.  The analog of \eqref{e:TDgoal} with $h_\alpha$ replaced by the relative value function can be solved using TD-learning techniques only for Markovian models that regenerate:  there exists a single state $x^\bullet$ that is visited infinitely often \cite{CTCN,huachemehmeysur11}.  

Regeneration is not a restrictive assumption in many cases.   However,  the asymptotic variance of the algorithms introduced in 
\cite{CTCN} grows with the variance of inter-regeneration  times.  The variance can be massive in simple examples such as the M/M/1 queue. High variance is also predominantly observed in the discounted cost case when the discounting factor is close to $1$.

The \textit{differential} TD-learning algorithms introduced in this paper are designed to resolve these issues.  The main idea is to estimate the  \emph{gradient} of the value function directly.  Under the conditions imposed, the variance remains uniformly bounded  over $0 < \alpha < 1$, and is also applicable for approximating the solution to Poisson's equation.

\subsection{Differential TD-learning}

In $\grad$-TD learning algorithms, the gradient of the value function is approximated rather than the function itself.    
\spm{dropped motivation here -- I think it is clear elsewhere}

Consider again the discounted-cost setting,  and suppose that both $h_\alpha$ and each of the potential approximations   $\{h_\alpha^\theta: \theta\in\Re^\ell\} $ are continuously differentiable ($C^1$) as a function of $x$.  In most of the paper, algorithms and analysis are restricted to the linear parameterization \eqref{e:hLinearPar}, so that  
\begin{equation}
 \grad h_\alpha^\theta = \sum_{j=1}^\ell \theta_j \grad \psi_j.
\label{e:gradhtheta}
\end{equation} 
The $\grad$-TD learning algorithm is designed to compute the solution to the following nonlinear program:
\begin{equation}
\theta^* =   \argmin_\theta \Expect[ \| \grad h_\alpha^\theta(X) - \grad h_\alpha(X)\|^2 ]\, ,\quad X\sim \pi\,.
\label{e:gradTD}
\end{equation}
Once the optimal parameter has been obtained,  the approximate value function requires the addition of a constant, 
\begin{equation}
 h_\alpha^{\theta^*} = \sum_{j=1}^\ell \theta_j   \psi_j + \kappa (\theta^*).
\label{e:gradTDaddConstant}
\end{equation}
The optimal choice of $\kappa(\theta^*)$ is also obtained in the algorithms described in this paper.
\notes{Odd? Can we write this as equal to something??
\\
spm:  I removed the "grad" before the psi.  I'm not sure what is odd here}

 In summary, the contributions of this work are,
 \begin{arabnum}
 	\item  
	\begin{alphanum} 
 \item
 The new
 $\grad$-TD algorithms
are applicable for either  discounted- and average-cost.
 \item
  For a linear parameterization,  the $\grad$-LSTD algorithm solves the quadratic program \eqref{e:gradTD}. 	
\item
 Extensions to nonlinear parameterizations are obtained.   
\end{alphanum}	
In the discounted-cost setting, the algorithms also compute the optimal constant  $\kappa(\theta^*)$ appearing in \eqref{e:gradTDaddConstant}.
	
	\item
	The new algorithms are applicable for models that do not have regeneration, and under general conditions the variance is uniformly bounded over all $0<\alpha<1$.
\end{arabnum} \notes{{Not talking about Poisson's equation in the list of contributions?}}

\spm{I don't think this  needs to be repeated here:
Also, as previously mentioned, in many applications, including policy iteration and optimal filtering, it is the gradient of the value function that is of interest rather than the value function itself. 
 }

\notes{they ARE only applicable
\\
spm: ?}

These algorithms do have limitations.  First, they are only applicable in settings where the gradient is a meaningful concept.  However, in numerical experiments we find that a pseudo-gradient can be defined even for a model with discrete state space, and the resulting ad-hoc algorithm is remarkably effective.   

Also,  in its current formulation,   the algorithms introduced here fall into the category of \emph{Approximate Dynamic Programming} (ADP), rather than \emph{Reinforcement Learning} (RL):  The algorithm relies on simulating a model of the system, rather than estimating a value function based on observations of a physical system.   This distinction is not absolute:  For example, in the applications to speed-scaling presented in \Section{sec:sim}, the $\grad$-LSTD learning algorithm does fall into the class of RL algorithms. 
 
 The remainder of the paper is organized as follows: basic definitions and value function representations are presented in \Section{s:Rep}.  The $\grad$-LSTD learning algorithm is introduced in  \Section{s:TD}.  Results from numerical experiments are surveyed in \Section{sec:sim}, and conclusions are contained in \Section{sec:conclusions}. 
 
 \spm{one gap in the intro -- we need to say something about Lyapunov exponents, no?}
 

\section{Representations and approximations}
\label{s:Rep}

We begin with assumptions,  and representations for value functions and their gradients.    

\subsection{Markovian model and problem formulation}
\label{s:setup}

The evolution equations  \eqref{e:SP_disc} define a Markov chain  $\bfmX$ with transition semigroup defined
for $t\ge 0$, $x\in\Re^d$, and  $A\in \clB(\Re^d)$, via
\[
P^t(x,A):=\Prob_x\{X(t)\in A\}:=\Pr\{X(t)\in A\,|\,X(0)=x\}.
\]
For $t=1$ we write $P=P^1$,  which has the following form: 
\[
P(x,A) = \Pr\{ a(x,N(1)) \in A\}.
\]
\notes{Was x + a(x,N)}

The first set of assumptions ensures that each of the value functions exists.  Fix a continuous function  $v\colon\Re^d\to [1,\infty)$ that serves as a weighting function. For any measurable function $f\colon\Re^d\to\Re$,  the $v$-norm is denoted by,
\[
\|f\|_v \eqdef \sup_x  \frac{|f(x)|}{v(x)} .
\]
The set of all measurable functions for which $\|f\|_v$ is finite is denoted $\Lv$.

\begin{Anum}{A1}
\item  
The Markov chain is \textit{$v$-uniformly ergodic}:  There exists a unique invariant probability measure $\pi$, $b_0 < \infty$,  and $0<\rho_0 < 1$,  such that for each function $f\in\Lv$,
\begin{equation}
\big| \Expect_x[f(X(t))] - \pi(f) \big| \leq b_0 \rho_0^t  \|f\|_v v(x) ,\quad t\ge 0\,,
\label{e:v-uni}
\end{equation}
where $\pi(f)$ denotes the steady-state mean of $f$.  
\end{Anum}
It is well known that (A1) is equivalent to the existence of a Lyapunov function: drift condition (V4) of \cite{CTCN}.
The following consequences are immediate:
\begin{proposition}
\label{t:v-uni-value} \notes{The following HOLDS?}
The following hold under (A1),  and the bound $\|c\|_v<\infty$: 
The limit  $\barc$  in \eqref{e:barc} exists, with $\barc = \pi(c) < \infty$, and is independent of the initial condition $x$. Moreover, there exists $b_c<\infty$ such that:
\[
\begin{aligned}
|h_\alpha(x)| & \leq b_c \big( v(x) + (1-\alpha)^{-1} \big)
\\
|h_\alpha(x) - h_\alpha(y) | & \leq b_c \big( v(x) +  v(y)\big)
\\
\text{\it
and} \qquad 
|h(x)|& \leq b_c v(x)\,, \qquad\qquad  x,y\in\Re^d
\end{aligned}
\] 
where $h_\alpha$ is defined in \eqref{e:DCOE_disc},  and $h$ is defined in \eqref{e:fish-sum}.
\bqed
\end{proposition}

The following operator-theoretic notation will simplify exposition.   
For any measurable function $f\colon\Re^d\to\Re$   the new function $P^t f$ is defined as the conditional expectation
\[
P^t f\, (x) = 
\Expect_x [f(X(t))  ]  \eqdef
\Expect [f(X(t)) \mid X(0) = x]  .
\]
The  \textit{resolvent kernel} is   
the ``$z$-transform'' of the semi-group,
\begin{equation}
R_\alpha \eqdef \sum_{t = 0}^{\infty} \alpha^t P^t,
\quad0<\alpha < 1.
\label{e:resolvent_disc}
\end{equation} 
Under the assumptions of \Prop{t:v-uni-value} we have  
\begin{equation}
\begin{aligned}
h_\alpha &= R_\alpha c\, .
\label{e:ResFormulah}
\end{aligned}
\end{equation} 
The solution to Poisson's equation has a similar representation under these assumptions -- an operator-theoretic representation of \eqref{e:fish-sum}
 \cite{MT}.



 \notes{Not talking about the uniqueness?  A: no need.
 \\
 Let's think about this:
 {With slight abuse of notation, we denote $h \equiv h_\alpha, \,\alpha =1$ in the rest of the paper.}
 \\
 We said we would start with discounting, let's just stick to that.}

The representation \eqref{e:ResFormulah} is valuable in deriving the  TD-learning algorithm \cite{bertsi96a,CTCN}.   We seek a similar representation for the gradient $\grad h_\alpha =  \grad [ R_\alpha c]$.      

\subsection{Representation for the gradient of a value function} 
\label{s:rep}

The goal here is to obtain an operator $\Omega_\alpha$ for which the following holds:
\begin{equation}
\grad h_\alpha = \Omega_\alpha \grad c.
\label{e:OmDream}
\end{equation}
This requires assumptions on the model and the cost function.   In the following, a heuristic construction of $\Omega_\alpha$ is presented, with justifications collected together in \Section{s:derLSTD}.   
For complete details, the reader is referred to \cite{devkonmey16a}.

The representation requires additional assumptions:     
 
\begin{Anum}{A2}
\item The disturbance process $\bfmN$ does not depend upon the initial condition $X(0)$.
\label{a:N}
\item The function $a$ is continuously differentiable in its first variable, with
\[
\sup_{x,n} \| \grad a (x,n)\| <\infty
\]
in which $\| \varble \|$ is any matrix norm, and the $i$th column of the $d\times d$ matrix  $\grad a$ is equal to $\grad a_i$. 
\spm{can give equation in journal}
\end{Anum}
The first assumption	A2.\ref{a:N} is critical so that the initial state $X(0)=x$ can be regarded as a   variable, with $X(t)$ being a continuous function of $x$.   Assumption~A2 allows us to define the \textit{sensitivity process} $\Sens(t)$:
\begin{equation}
\Sens_{i,j} (t) \eqdef \frac{\partial X_i(t)}{\partial X_j(0)}, \quad 1\leq i,j \leq d.
\label{e:SensDef}
\end{equation}
From \eqref{e:SP_disc}, the sensitivity process evolves according to the random linear system
\spm{the transpose seems odd.  Why did we introduce this?}
\begin{equation}
\begin{aligned}
\clS(t+1) &= \dffA^\transpose(t+1) \clS(t) ,\quad \clS(0)=I\,,
\end{aligned}
\label{e:SP_sens_disc}
\end{equation}
where $\dffA(t) \eqdef \grad a\, (X(t-1),N(t))$.

The operator $\tilgrad$ is defined for  any $C^1$ function $f  $ via 
\begin{equation}
\tilgrad f (X(t)) \eqdef \Sens^\transpose(t) \grad f(X(t))\, .
\label{e:tilgrad}
\end{equation}
It follows from the chain rule that this coincides with the gradient of $f(X(t))$ with respect to the initial conditions.

\spm{no room:
\[
\tilgrad f (X(t)) = \Bigl[\frac{\partial  f(X(t))}{\partial{X_1(0)}}, \ldots, \frac{\partial f(X(t))}{\partial{X_d(0)}} \Bigr]^\transpose.
\]
}

The interpretation of \eqref{e:tilgrad} motivates the introduction of a semi-group $\{Q^t\}$ of operators, whose domain includes functions $g: \Re^d \to \Re^d$ for which $g_i\in\Lv$ for each $i$.    For $t=0$,  it is the identity operator,
and for $t\ge 1$, 
\begin{equation} 
Q^t g(x) \eqdef \Expect_x \bigl[ \clS^\transpose(t) g(X(t))  \bigr].
\label{e:Qt} 
\end{equation}
Provided we can exchange the gradient and the expectation,  
\[ 
\frac{\partial}{\partial x_i} \Expect[f(X(t)) ]  = \Expect \bigl[ [\tilgrad f(X(t))]_i  \bigr]
\]
which implies that $  \grad P^t f(x)  = \Expect_x [\tilgrad f(X(t))]  = Q^t \grad f\,(x)$.
Under further assumptions we obtain
\[
\grad h_\alpha = 
\sum_{t=0}^{\infty} \alpha^t \grad P^t c 
=
\sum_{t=0}^{\infty} \alpha^t Q^t \grad c\, .
\]
We then obtain \eqref{e:OmDream}, with
\begin{equation} 
\Omega_\alpha g \eqdef \sum_{t=0}^{\infty} \alpha^t Q^t g.
\label{e:Omegarep_disc}   
\end{equation}  
The representation  \eqref{e:OmDream} is the basis of the  $\grad$-LSTD~learning algorithms developed in this paper.

Under the conditions of \Prop{t:dLSTDconverges} the operator $\Omega_\alpha$ is well defined on a large domain of functions, with uniform bound over all $0\le \alpha \le 1$.   In the special case of $\alpha =1$ we denote  $\Omega=\Omega_1$, which under these conditions provides the representation  $\grad h = \Omega \grad c$ for the gradient of the relative value function.

\spm{Tighten up once "theorem mathy" is clarified.
Also, $Z$ is not defined in CDC,
$ \grad Z \tilc =$
}


\spm{keep section titles short... for a discrete time dynamical system}

\section{Differential TD~learning}
\label{s:TD}

Algorithms are developed here for the Markov model \eqref{e:SP_disc}, subject to Assumptions~A1 and A2.
The algorithms are presented first,  with supporting theory postponed to \Section{s:derLSTD}. 

Until \Section{s:extend} we restrict to a  linear parameterization, in which $h_\alpha^\theta = \theta^\transpose\psi$.

\subsection{LSTD algorithms}

We begin with a review of
the standard algorithm, which is defined by the following recursion.

\noindent{\bf Least squares TD-learning algorithm}
\begin{equation}
\begin{aligned}
\varphi(t) & = \alpha \varphi(t-1) +  \psi(X(t))
  \\
  b(t) & = (1-\gamma_t)  b(t-1) + \gamma_t \varphi(t) c(X(t))
\\
  M(t) & = (1-\gamma_t)  M(t-1) + \gamma_t \psi(X(t)) \psi^\transpose(X(t)),
 \label{e:TD2}
\end{aligned}
\end{equation}
and obtain $\theta(t) = M^{-1}(t) b(t)$.
The algorithm is initialized with $b(0)$, $\varphi(0) \in\Re^\ell$,  and   $M(0)>0$   a positive-definite $\ell\times\ell$ matrix  \cite{bertsi96a,CTCN}.

Throughout the paper the gain sequence appearing in  \eqref{e:TD2} and elsewhere is taken to be $\gamma_t = 1/t$, $t\ge 1$. 
 
\spm{We use this term later, but awkward here:
$\varphi_\psi(t)$ is called the \emph{eligibility vector}}

To simplify discussion we restrict to a stationary setting in the convergence results in this paper:
\begin{proposition}
\label{t:LSTDconverges}
Suppose that (A1) holds, and that $c^2$ and $\|\psi\|^2$ are in $\Lv$.   Suppose moreover that the 
matrix $M=\Expect_\pi[\psi(X) \psi(X)^\transpose]$  is full rank, where $X\sim \pi$.  Then there is a version of the pair process $(\bfmX,\bfvarphi)$ that is stationary.   For any initial conditions $b(0)  \in\Re^\ell$  and   $M(0)>0$,  the algorithm is consistent:
\[
\theta^* = \lim\limits_{t\to\infty} M^{-1}(t) b(t).
\]
\end{proposition}

\begin{proof}
The existence of a stationary solution $\bfmX$ follows directly from $v$-uniform ergodicity, and we then  define 
\[
 \varphi(t) =  \sum_{i=0}^\infty \alpha^i   \psi(X(t-i)).
\]
It is known that the optimal parameter can be expressed 
$\theta^* =  M^{-1} b$ in which   $b =\Expect_\pi[\varphi(t) c(X(t))]$, so the result follows from the Law of Large Numbers for stationary processes.
\end{proof}
\notes{{Expectation with respect to $\pi$ no??}}

In the construction of the LSTD algorithm, the optimization problem \eqref{e:TDgoal} is cast in the Hilbert space,
\[
 L_2^\pi = \bigl\{ \text{measurable } h\colon\Re^d\to\Re \ : \ \|h\|_\pi^2 = \langle h,h \rangle_\pi <\infty\bigr\}
 \]
with
 $
\langle f,g \rangle_\pi \eqdef  \int  f(x)g(x) \pi(dx)$.  
The $\grad$-LSTD algorithm is based on a different inner product to define the norm in the approximation error.

For $C^1$ functions $f$, $g$ for which $\| \grad f\|, \|\grad g\| \in L_2^\pi$,
 define the inner product
\[
\langle f,g \rangle_{\pi,1} = \int  \grad f (x)^\transpose \grad g (x) \pi(dx), 
\]
with the associated norm   $\| f\|_{\pi,1} = \sqrt{\langle f,f \rangle_{\pi,1} }$.   The nonlinear program \eqref{e:gradTD} can be recast  as 
\begin{equation}
\theta^* = \argmin_{\theta} \|h_\alpha^\theta - h_\alpha \|_{\pi,1}.
\label{e:OptThetaGhalphaDT}
\end{equation}


Consider the linear parameterization \eqref{e:gradhtheta} in which $\psi\colon\Re^d\to\Re^d$ is continuously differentiable, and assume as well that $c$ is continuously differentiable. The $\grad$-LSTD~learning algorithm is then defined by the following recursion: 

\noindent{\bf Differential least squares TD-learning algorithm}
\begin{equation}
\begin{aligned}
 \varphi(t) & = \alpha \dffA^\transpose(t) \varphi(t-1) +  \grad \psi(X(t))
\\
b(t) & = (1-\gamma_t) b(t-1) + \gamma_t \varphi(t)^\transpose \grad c(X(t))
	\\
	\!\!\!
M(t) & = (1-\gamma_t) M(t-1) + \gamma_t \grad \psi(X(t)) \grad \psi(X(t))^\transpose
\label{e:dTD1}
\end{aligned}
\end{equation}
where $\grad \psi\, (x)$ denotes the $d \times \ell$ matrix
\begin{equation}
[ \grad \psi \, (x)]_{i,j} = \frac{\partial }{\partial x_i} \psi_j (x)\,, \quad x\in\Re^d,
\label{e:gradpsi}
\end{equation}
and the parameters are obtained as, 
\spm{no transpose in diffA, right?}
$
\theta(t)=  M^{-1}(t) b(t)$.  
Once again, $M(0)>0$ is an arbitrary $\ell\times\ell$ positive-definite matrix, $b(0)\in\Re^\ell$, $\varphi(0) \in\Re^{d \times \ell}$ are arbitrary initializations.	 

Two more steps are required to obtain an estimate of $h_\alpha$.   To ensure that $\pi(h_\alpha) = \pi( h_\alpha^{\theta})$ we   take $
h_\alpha^{\theta} = \theta^\transpose \psi + \kappa (\theta)$ where the constant is
\[
\kappa (\theta) = -\pi(h_\alpha^{\theta}) + \pi(c) / (1-\alpha).
\]
The two means can be estimated recursively: 
\notes{Important: Remove eqnarray here, and everywhere else, or only for the algorithms?
spm:  avoid eqnarray when possible, but we can worry about this in the summer for the final draft.  Also, don't use equation numbers unless you think you will reference the equation.}
\begin{eqnarray} 
		\barh_\alpha(t) & = & (1-\gamma_t) \barh_\alpha(t-1) + \gamma_t  h_\alpha^{\theta(t)}
 	\label{e:dTD4}
	\\	
\displaystyle
		\barc(t) & = & (1-\gamma_t) \barc(t-1) + \gamma_t c(X(t)).
 	\label{e:dTD5}
\end{eqnarray} 
It is immediate that $\barc(t) \to \barc$ as $t\to\infty$ by the Law of Large Numbers for $v$-uniformly ergodic Markov chains \cite{MT}.   Convergence of   $\barh_\alpha(t)$ to $\pi(h_\alpha^{\theta^*}) $ requires further assumptions.
 
This completes the description of the   $\grad$-LSTD~learning algorithm.

\subsection{Derivation and analysis}  
\label{s:derLSTD}
 
For a linear parameterization, the optimal parameter is the minimum of a quadratic.
The proof of \Prop{t:LSTDquad} follows immediately  from the definition of the norm $\|\cdot\|_{\pi,1}$.
\begin{proposition}
\label{t:LSTDquad}
The norm appearing in \eqref{e:OptThetaGhalphaDT} is a quadratic form,    
\begin{equation}
\|h_\alpha^\theta - h_\alpha\|^2_{\pi,1} = \theta^\transpose M \theta - 2b^\transpose\theta + k ,
\label{e:QuadtraticRep}
\end{equation}
in which for each $1\le i, j\le j$,
\begin{equation}
M_{i,j} = \langle \psi_i, \psi_j \rangle_{\pi,1}, \quad b_i = \langle \psi_i,  h_\alpha \rangle_{\pi,1},
\label{e:bMInitDef}
\end{equation}
and $k = \langle h_\alpha,h_\alpha \rangle_{\pi,1}$.  Consequently, the optimizer \eqref{e:OptThetaGhalphaDT}
is any solution to 
\begin{equation}
M \theta^* = b.
\label{e:OptimalTheta}
\end{equation}
\bqed
\end{proposition}

The matrix $M$ can be expressed in more compact notation:
\begin{equation}
M = \Expect_\pi [(\grad \psi(X))^\transpose \grad \psi(X)].
\label{e:MDef}
\end{equation}
where $\grad \psi$ is defined in \eqref{e:gradpsi}. Similarly,   
\begin{equation}
b = \Expect_\pi \bigl[ [ \grad \psi(X) ]^\transpose \grad h_\alpha(X) \bigr].
\label{e:bDef}
\end{equation} 

As in the standard TD learning algorithm, the vector $b$ is represented using the function $h_\alpha$, which is unknown.  An alternative representation can be obtained whenever   \eqref{e:OmDream} is valid, and this is the basis of the $\grad$-LSTD algorithm.  

The following assumptions are used to justify this representation:


\noindent
{\bf{Assumption~{A3}}:}  For   any $C^1$ functions $f,g$ satisfying $f^2, g^2\in\Lv$ and $\|\grad f\|^2,\|\grad g\|^2\in\Lv$, 
 the following hold for  the stationary version of the Markov chain:
\begin{eqnarray}
\sum_{t=0}^{\infty}    \Bigl|  \Expect_\pi \bigl[ \grad f(X(t))\big)^\transpose    \Sens(t)   \grad g(X(0))  \bigr] \Bigr|  &<& \infty
\label{e:A3_true}
\\[.2cm]
\sum_{t=0}^{\infty}   \Expect_\pi  \Bigl[  \bigl| \grad f(X(t))\big)^\transpose    \Sens(t)   \grad g(X(0)) \bigr| \Bigr]  &<& \infty  
\label{e:A3_dream}
\end{eqnarray}

Under \eqref{e:A3_true}  the right hand side of \eqref{e:OmDream} is well defined a.e.\ $[\pi]$
when $c$ satisfies these conditions.   General conditions for the validity of \eqref{e:A3_true} 
 are established in  \cite{devkonmey16a}.   Theory to justify \eqref{e:A3_dream}
 is not as well developed.   The condition is related to the existence of a negative Lyapunov exponent.
 

Under these assumptions we obtain a stationary solution for the pair $(\bfmX,\bfvarphi)$.  The representation of $\bfvarphi$ requires the following shift-operator on sample space for a stationary version of $\bfmX$: For a random variable of the form $Z = F(X(r),N(r), \dots, X(s), N(s))$ with $r\le s$ we denote, for any integer $k$,
\[
\Theta^k Z = F(X(r+k),N(r+k), \dots, X(s+k), N(s+k))
\]
Consequently,  
\begin{equation}
\Theta^k \Sens(t) =    [\dffA(1+k) \dffA(2+k)\cdots  \dffA(t+k) ]^\transpose\, .
\label{e:shift}
\end{equation}

\Lemma{t:gradLSTDstat} follows immediately from the assumptions:  It follows from
the definition \eqref{e:EliVecDT} that this process follows the first recursion in \eqref{e:dTD1}.

\begin{lemma}
\label{t:gradLSTDstat} 
Suppose that Assumptions~A1--A3 hold, and that   $\|\psi\|^2$ and  $\|\grad \psi\|^2$ are in $\Lv$.    Then there is a version of the pair process $(\bfmX,\bfvarphi)$ that is stationary, with
\begin{equation}
\varphi(t) = \sum_{k=0}^{\infty} \alpha^k \big[ \Theta^{t-k} \Sens(k) \big] \gradpsi(X(t-k)), \quad t\in\ZZ\, .
\label{e:EliVecDT}
\vspace{-0.15in}
\end{equation} 
\bqed
\end{lemma}

\begin{proposition}
\label{t:dLSTDconverges}
Suppose that Assumptions~A1--A3 hold, and that $c^2$, $\| \grad c\|$,  $\|\psi\|^2$ and  $\|\grad \psi\|^2$ are in $\Lv$. Suppose moreover that the matrix $M$ appearing in \eqref{e:MDef} 
  is full rank. Then, for the stationary  process $(\bfmX,\bfvarphi)$, 
 the $\grad$-LSTD learning algorithm is consistent: for any initial $b(0)  \in\Re^\ell$  and   $M(0)>0$,  
\[
\theta^* = \lim\limits_{t\to\infty} M^{-1}(t) b(t).
\]
Moreover, with probability one,
\[
\barc=
\lim_{t\to\infty}  \barc(t) ,\quad \pi(h_\alpha^{\theta^*}) = 
\lim_{t\to\infty} \barh_\alpha(t) ,
\]
and hence $\lim_{t\to\infty} \{
-\barh_\alpha(t)  +  \barc(t) /(1-\alpha)\} =\kappa(\theta^*)$.
\end{proposition}
 
The remainder of this section consists of a proof of this proposition.  We begin with a representation of $b$:

\begin{lemma}
\label{e:bRep}
Under the assumptions of \Prop{t:dLSTDconverges},  
\begin{equation}
\begin{aligned}
b^\transpose &= \sum_{t=0}^{\infty} \alpha^t    \Expect_\pi \bigl[ \bigl ( \Sens^\transpose(t) \grad c(X(t)) \bigr)^\transpose \gradpsi(X(0))   \bigr]    
\\
&=  \Expect \Bigl [ \big(\grad c(X(0))\big)^\transpose \varphi(0) \Bigr ]\, .
\end{aligned}
\label{e:bpsi_disc_init}
\end{equation}
in which $\bfmX$ is stationary, with marginal $\pi$.
\end{lemma}

\begin{proof}
The representation \eqref{e:OmDream} is valid under (A3).  
Using this and \eqref{e:SP_sens_disc} gives the first representation in \eqref{e:bpsi_disc_init}:
\begin{equation}
\begin{aligned}
b^\transpose  &= \int \Expect_x \big [(\Omega_\alpha \grad c(x))^\transpose \gradpsi(x) \big ]\pi(dx) 
\\
&= \sum_{t=0}^{\infty} \alpha^t  \int \Expect_x \big [( \Sens^\transpose(t)  \grad c(x))^\transpose \gradpsi(x) \big ]\pi(dx)
\\
&= \sum_{t=0}^{\infty} \alpha^t    \Expect_\pi \bigl[ \bigl ( \Sens^\transpose(t) \grad c(X(t)) \bigr)^\transpose \gradpsi(X(0))   \bigr]   
\label{e:bpsi_disc_init2}
\end{aligned}
\end{equation}

\spm{Is this useful?
\\
Note that each $\dffA(s)$ is a $d \times d$ matrix, $\grad c (X(t))$ is a $d \times 1$ column vector, and as mentioned earlier $\grad \psi$ is a $d \times \ell$ matrix. Therefore, the right hand side is a $1 \times \ell$ row vector, implying that $b$ is a $l \times 1$ column vector.}

  \notes{{Proposition to formalize stationarity? What should the proposition be?  
\\
\bf I am pretty sure the expression below is false}
\\
I think these concerns are now resolved}

Stationarity implies that for any $t,k\in\ZZ$,
\begin{equation*}
\begin{aligned}
\Expect_\pi \Bigl[ \Bigl (&\Sens^\transpose(t)  \grad c(X(t)) \Bigr)^\transpose \gradpsi(X(0)) \Bigr] 
\\
&=  \Expect_x \Bigl[ \Bigl ( [\Theta^k\Sens^\transpose(t) ]\grad c(X(t+k)) \Bigr)^\transpose \gradpsi(X(k)) \Bigr].
\end{aligned}
\end{equation*} 
Setting $k=-t$, the first representation in \eqref{e:bpsi_disc_init} becomes:
\[
\begin{aligned}
\hspace{-.5cm}
b^\transpose  &= \sum_{t=0}^{\infty} \alpha^t \Expect_x \Bigl[ \big(\grad c(X(0))\big)^\transpose \big( \Theta^{-t} \Sens(t) \big)  \gradpsi(X(-t)) \Bigr]
 \\
&= \Expect_x \Bigl[\big(\grad c(X(0))\big)^\transpose \Bigl(  \sum_{t=0}^{\infty} \alpha^t  \big( \Theta^{-t} \Sens(t) \big) \gradpsi(X(-t)) \Bigr) \Bigr]\, .
\end{aligned}
\]
The last equality   is obtained under Assumption~A3 by applying Fubini's theorem,  and this completes the proof.
\end{proof}

\paragraph*{Proof of \Prop{t:dLSTDconverges}}   \Lemma{e:bRep} combined with the stationarity assumption implies that
\[
\begin{aligned}
\lim_{T\to\infty } \frac{1}{T} b(t) &=
\lim_{T\to\infty } \frac{1}{T} \sum_{t=1}^T \varphi(t)^\transpose \grad c(X(t))  
\\
& = \Expect[\varphi(0)^\transpose \grad c(X(0)) ] =b 
\end{aligned}
\]
Similarly,  for each $T\ge 1$ we have 
\[
M(T) = M(0) + \sum_{t=1}^{T} (\gradpsi(X(t)))^\transpose \gradpsi(X(t)),
\]
and by the Law of Large Numbers we once again obtain 
\[
\lim_{T\to\infty } \frac{1}{T}  M(T) = M. 
\]
Combining these results establishes  $\theta^* = \lim\limits_{t\to\infty} M^{-1}(t) b(t)$.

Convergence of $ \{\barc(t)\}$ is identical,  and convergence of $ \{\barh_\alpha(t)\}$ also follows from the Law of Large Numbers since we have convergence of $\theta(t)$.  
\bqed

\subsection{Extensions}
\label{s:extend}
 
 \subsubsection*{Extension to average-cost}    Nowhere in the proof   of \Prop{t:dLSTDconverges} do we use the assumption that $\alpha<1$.   It is not difficult to establish that under the conditions of the proposition,  the $\grad$-LSTD~learning algorithm is convergent when $\alpha=1$,  and the limit solves the quadratic program
\[
 \theta^* = \argmin_{\theta} \|h^\theta - h \|_{\pi,1}
\]
in which $h$ is a solution to Poisson's equation.   

\subsubsection*{Nonlinear parameterization}  If the parameterized family $\{ h^\theta_\alpha\}$ is nonlinear in $\theta$,
 then the optimization problem \eqref{e:OptThetaGhalphaDT} may not be convex.   It is possible to construct algorithms to compute a local minimum through stochastic gradient techniques.

The basis should be chosen so that $h^\theta$ and $\gh_\alpha^\theta$ are continuous function of $x$, and continuously differentiable in $\theta$.  On denoting $\psi_i^\theta = {}$ the partial derivative of $ h_\alpha^\theta$ with respect to $\theta_i$,
the first order condition for optimality of \eqref{e:OptThetaGhalphaDT} is,
\spm{eqn cut for CDC}
\[ 
0 = \Expect_\pi [   (\grad \psi^{\theta^*}(X))^\transpose(\grad h_\alpha^{\theta^*}(X) - \grad h_\alpha(X))  ]\,,
\]
where the $i$th column of the gradient matrix $\grad \psi^{\theta^*}$ is equal to $ \grad \psi_i^{\theta^*}$.

The inner product $ \langle   h_\alpha, \psi_i^{\theta^*} \rangle_{\pi,1} $ depends on the unknown function $h_\alpha$, but this can be transformed into a practical algorithm.   For example,  a gradient descent like stochastic approximation algorithm is defined through the recursions,
\[
\begin{aligned} 
  \varphi(t)  &= \alpha \varphi(t-1) +  \psi^{\theta(t-1)}(X(t))
\\
d(t) & =  \grad \psi^{\theta(t-1)}(X(t))^\transpose \grad h_\alpha^{\theta(t-1)}(X(t)) -   \varphi(t)^\transpose \grad c(X(t))   
\\
\theta(t) & = \theta(t-1) -   \gamma_t  \grad \psi^{\theta(t-1)}(X(t)) d(t).
\end{aligned} 
\]
Stability analysis of these algorithms will be the subject of future research.   
 
Extensions to Markov models in continuous time are contained in the online version of the paper \cite{devkonmey16a}.

\section{Continuous time}
\label{sec:CT}

To highlight the main ideas, we restrict to a one-dimensional diffusion on $\Re$,  
 \spm{no square before Brownian motion.   Also, the subscript $x$ looked like $X(0)=x$}
\begin{equation}
 dX(t) = a(X(t)) \,dt + \sigma(X(t)) \, dB(t),
\label{e:SDE}
\end{equation}
in which $\bfmB$ is standard Brownian motion, and the function $a\colon\Re\to\Re$ is Lipschitz continuous.  For simplicity we also take $\sigma(x)\equiv 1$.  Its semigroup is denoted $\{P^t\}$, and its differential generator is defined for $C^2 $ functions $f\colon\Re\to\Re$ via,  $\generate f = a f' + \half f''$.  Assumption~A1 is imposed, where again the ergodic limit \eqref{e:v-uni} is equivalent to the existence of a Lyapunov function, defined with respect to the differential generator.
\spm{really should reference dowmeytwe or ...}

The discounted cost is based on a discount rate $\gamma>0$, with definition similar to 
 \eqref{e:DCOE_disc}:
\begin{equation*}
h_\gamma(x)=  \int_{t}^{\infty} e^{-\gamma t}  \Expect_x [c(X(t)) ]\, .
\label{e:DCOE_disc_cts}
\end{equation*}
The   {\em resolvent kernel} $R_{\gamma}$ is the Laplace transform,
\begin{equation}
R_{\gamma} \eqdef   \int_0^\infty  e^{-\gamma t} P^t\, dt ,
\quad\gamma>0\, ,
\label{e:resolvent_cont}
\end{equation}
so that the value function can be expressed $h_\gamma = R_\gamma c$. 

If the value function is $C^2$, then it solves the dynamic programming equation,
\begin{equation}
\clD h_\gamma = \gamma h_\gamma - c
\label{e:DCOEcts}
\end{equation}
A representation for the derivative $h_\gamma'$ is obtained in the following subsection, from which we obtain a continuous-time analog of the $\grad$-LSTD algorithm.


 \subsection{Gradient representation}

The dynamic programming equation \eqref{e:DCOEcts} suggests the inverse formula $R_\gamma = [I_\gamma -\clD]^{-1}$.  This formula is valid on a suitable domain,  and with suitable interpretation \cite{meytwe93e}.

The representation for $h_\gamma'$ makes use of the the generalized resolvent kernel  \cite{nev72,meytwe93e,devkonmey16a}:
For a measurable function $G\colon\Re\to\Re$, and measurable functions $f$ in some domain,
 \begin{equation}
 R_G f\, (x) \eqdef \int_0^\infty \Expect_x\Bigl[ \exp\Bigl(-\int_0^t G(X(s))\, ds \Bigr) f(X(t))\Bigr]\, dt.
\label{e:Neveu}
\end{equation}
In  \cite{nev72,meytwe93e} it is assumed that $G>0$ everywhere. These conditions are relaxed in \cite{konmey03a,devkonmey16a}.     

\spm{what conditions??
Under these conditions, the formula $R_G = [I_G -\clD]^{-1}$ is valid on some domain.
}
 
To apply these concepts,  differentiate with respect to $x$ each side of \eqref{e:DCOEcts} to obtain
 \[
 a' h_\gamma' + a h_\gamma '' + \half h_\gamma '''
= \frac{d}{dx}  (\clD h_\gamma) = \gamma h_\gamma' - c'
\]
Rearranging terms gives $[I_G -\clD] h_\gamma' = c'$, with $G=- a' + \gamma$.  Provided the inverse exists, we obtain
\begin{equation}
 h_\gamma' = [I_G - \generate]^{-1} c' = R_G c'\, .
\label{e:gradhrep}
\end{equation}
These steps can be justified subject to a growth condition on $c'$,  and a Lyapunov drift condition for the diffusion \cite{devkonmey16b}.
 
 \spm{
Note that any diffusion process with negative drift satisfies this property, and as long as $a' < \gamma$, the representation \eqref{e:gradhrep} is valid for $0 \leq \gamma < 1$, wherein $h_0'$ represents the derivative of the solution to Poisson's equation. 
\\
 we don't have to do that!!}
 
\subsection{$\grad$-LSTD-learning}
 
The goals are unchanged in this continuous time setting:  We seek the parameter $\theta^*$ that solves
 \begin{equation}
 \theta^* = \argmin_\theta \| h_\gamma^\theta - h_\gamma \|^2_{\pi,1}
\label{e:OptThetaCon}
 \end{equation}
The $\grad$-LSTD-learning algorithm designed to solve this problem is defined by these ODEs: 
\begin{subequations}
\begin{eqnarray}
\ddt
 \varphi(t) & = & [a'(X(t)) - \gamma] \varphi(t) + \psi'(X(t)) 
 	\label{e:dTD3_cts}
\\
\ddt 
b(t) & = &    \varphi(t)   c'(X(t)) 
\label{e:dTD2_cts}
\\
\ddt M(t) & = &    \psi'(X(t))   {\psi'}^\transpose(X(t))
\label{e:dTD1_cts}
\end{eqnarray}
\end{subequations}
generating estimates of $\theta^*$ as before via $\theta(t) = M(t)^{-1} b(t)$.

\smallbreak

The construction and analysis of this algorithm is based on the characterization of $\theta^*$ for the linearly parameterized approximation. We close this section with an overview of the main ideas.

As in the discrete time case, the objective is quadratic in $\theta$,   and 
as in \Prop{t:LSTDquad} we obtain $\theta^*= M^{-1} b$ with  $M$ defined in \eqref{e:MDef},  and $b = \Expect_\pi \bigl[ \psi'(X)    h_\gamma'(X) \bigr]$.   

The main difference in the   continuous time case is the alternate representation for $b$.  We again require assumptions to ensure the existence of a steady-state solution to  	\eqref{e:dTD3_cts}, of the form
\[
 \varphi(t)
 =
  \int_{-\infty}^t  \exp\Bigl(-\int_r^t G(X(s))\, ds  \Bigr) \psi'(X(r))   \,  dr,\quad t\in\Re.
\]
with $G(x)= \gamma-a'(x)$, $x\in\Re$.  This requires a version of Assumption~A3 in the continuous time setting;  in \cite{devkonmey16b} it is shown that this holds under a Lyapunov drift condition. 

 When these steps are justified we can conclude that  $b = \Expect_\pi \bigl[ [  \varphi(t)   c'(X(t)) \bigr]$,  and convergence of the $\grad$-LSTD algorithm then follows as in the discrete-time setting.

 \section{Simulation Results}
\label{sec:sim}

This section contains a survey of numerical experiments to illustrate the general theory,  and suggest possible
extensions of the algorithm.


Common elements in all of our experiments are a linear parameterization for the value function, and the implementation of the $\grad$-LSTD algorithm.  Comparisons with other approaches include 
the standard LSTD algorithm for discounted cost, and the regenerative LSTD algorithm of \cite{CTCN,huachemehmeysur11} for average cost applications where there is regeneration.  
The standard TD($\lambda$) algorithm was also considered, but in each example the variance was found to be  several orders of magnitude greater than alternatives.   A matrix gain variant called TD-$K$($\lambda$) is introduced to obtain a better algorithm for comparison.  For this linearly parameterized setting, the matrix gain algorithm is essentially equivalent to the LSTD($\lambda$) algorithm of \cite{boy02}.

The asymptotic covariance is used to compare these algorithms:   
\begin{equation}
\Sigma = \lim_{t\to\infty} \frac{1}{t} \Expect[\tiltheta(t)\tiltheta(t)^\transpose]
\label{e:aVar}
\end{equation}
where $\tiltheta(t)=\theta(t)-\theta^*$.  Under general conditions, the asymptotic covariance   coincides with the covariance in the Central Limit Theorem.  It is estimated by observing a histogram following multiple runs of each algorithm.

\spm{for journal: The existence of a finite limit requires conditions on the ODE associated with the stochastic approximation algorithm cite{whattodo-Borkarhasincorrectformula}.}

Two extensions are considered for a specific example: the approximation of relative value function for the speed-scaling model of \cite{chehuakulunnzhumehmeywie09}.  First, for this reflected process evolving on $\Re_+$, it is shown that the sensitivity process can be defined, subject to conditions on the dynamics near the boundary.   Second,  the algorithm is tested for an example with \textit{discrete state space}.  There is no apparent justification for this approach, but it worked well in the examples considered.

\subsection{Linear stochastic process}

A scalar linear model is ideal for illustrating the difference between $\grad$-LSTD learning and alternative approaches.   The dynamics are given by the recursion 
\[
X(t+1) = a X(t) + N(t+1)
\]
in which $a \in (0,1)$ and $\bfmN$ is Gaussian $\clN(0,1)$.

In all of the numerical results surveyed here, the cost function is defined to be the quadratic,  $c(x)=x^2$,   $a=0.7$,  and we restrict to the discounted-cost.   

The relative value function and the discounted-cost value functions are quadratic in this case, and also symmetric:  $h(x)=h(-x)$. Consequently, the function class obtained using the basis $\psi(x) = (1,x^2)^\transpose$ includes the true value function.  
On taking the gradient, the function space collapses from two dimensions to one,  and the growth rate of the functions is reduced from quadratic to linear: $\grad\psi(x) = (0,2x)^\transpose$. 
\spm{shouldn't we mention $\kappa^*(\theta)$ somewhere here?}

\spm{I'm so glad you went with bound $b$ and $M$ definitions. 
Could you revise v5 to follow this form?
\\
Note that I reversed the order:  we need varphi to come before b.   Make sure we stick to this convention everywhere}

\notes{Do we even need the following? All we are saying is $\dffA = a$ in this case. We can save space.
\\
spm: revised slightly.
}
For this linear model, the first recursion for the $\grad$-LSTD algorithm defined in \Section{s:TD} becomes
\begin{equation}
\begin{aligned}
\varphi(t) & = \alpha a \varphi(t-1) +  \grad \psi(X(t)),
\vspace{-0.3in}
\end{aligned}
\label{e:dLSTDLinear}
\end{equation}  
while the corresponding equation in the standard LSTD algorithm is
\vspace{-0.12in}
\begin{equation}
\begin{aligned}
\varphi(t) & =  \alpha \varphi(t-1) +    \psi(X(t)).
\vspace{-0.15in}
\end{aligned}
\label{e:LSTDLinear}
\end{equation}


Both these algorithms are consistent.  However, two differences suggest that the asymptotic covariance is much smaller when using the $\grad$-LSTD algorithm.  First is the
additional   discounting factor $a$  appearing in \eqref{e:dLSTDLinear}, but absent in  \eqref{e:LSTDLinear}.   This is why the LSTD asymptotic covariance grows without bound as $\alpha$ tends to $1$.   A second advantage of the 
$\grad$-LSTD algorithm is that the gradients reduce the growth rate of each function of $x$.  In this case, reducing the quadratic growth of $c$ and $\psi$ to the linear growth of their gradients.

Experiments were run for two different discounting factors, $\alpha=0.9$ and $\alpha=0.99$.  Variance estimates were obtained by conducting $10^3$ independent simulations for each set of parameters tested. 

\begin{figure}[h]
\Ebox{.75}{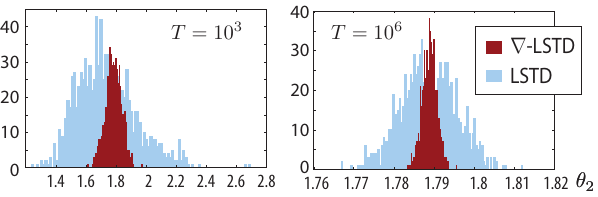}
\caption{Histogram of $\theta_2(T)$ using both TD-learning and $\grad$TD-learning, $\alpha=0.9$.}
\label{f:HistLinearAllp9}	
\end{figure}


The optimal parameter is $\theta^*=(16.1, 1.79)^\transpose$ when $\alpha=0.9$. \Figure{f:HistLinearAllp9}	
 shows the resulting histograms for $\theta_2(T)$ (the coefficient of $\psi_2(x) = x^2$) for two time horizons, $T=10^3$ and $10^6$.
 
It was found that convergence of the $\grad$-LSTD-learning algorithm is about $10$ times faster than the LSTD algorithm. That is, for a given time $T$, the variance of $\theta_2$ estimated using $\grad$TD-learning algorithm is about the same as that of the TD-learning algorithm which has run for $10$ times longer.

\notes{{Let's discuss, but for now let me say this.. \\
		Here we have a parameterization that includes the true $h^*$.   Hence TD(lambda) is consistent (Galerkin is the key to validation). \\
		If the parameterization does not include the true $h^*$ then there may indeed be bias, but we have to discuss the definition of bias.}}

The difference in performance of the two algorithms is greater as $\alpha$ is increased.   The case $\alpha=0.99$ is considered next, for which $\theta^*=(192.27, 1.9421)^\transpose$.   \Figure{f:HistLinearAllp99} contains histograms for the same two time horizons.   

In conclusion, for this example, the asymptotic covariance of the $\grad$-LSTD algorithm is bounded uniformly over $0<\alpha<1$, and it can also be used to estimate the solution to Poisson's equation.

\begin{figure}[h]
 \Ebox{.75}{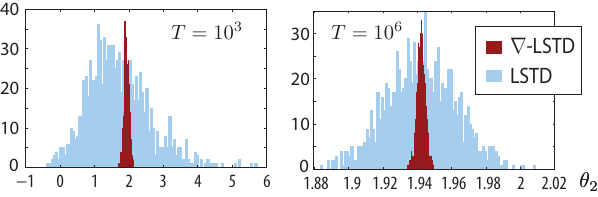}
		\caption{Histogram of $\theta_2(T)$ using both TD-learning and $\grad$TD-learning, $\alpha=0.99$.}
		\label{f:HistLinearAllp99}	
\end{figure}


\subsection{Dynamic speed scaling}

Dynamic speed scaling is a popular approach to power management in computer system design. The goal   is to control the processing speed so as to optimally balance energy and delay costs; 
this can be done by reducing (increasing) the processor speed at times when the workload is small (large). For the purposes of this paper, speed scaling is a simple stochastic control problem: a single server queue with a controllable service rate.

A regenerative form of LSTD learning was applied in \cite{chehuakulunnzhumehmeywie09} for this example to approximate the solution to the average-cost optimality equation.   Approximate policy iteration algorithm was implemented, in which the LSTD algorithm provided an approximate relative value function at each iteration of the algorithm.  

The discrete time MDP (Markov Decision Process) model is described as follows:   At each time $t$,  the state $X(t)$ is interpreted either as the queue length, or the workload in the system;   $N(t) \geq 0$ denotes the number of job arrivals,  and $U(t)$ the service completion at time $t$.  This is subject to the constraint $0 \leq U(t) \leq X(t)$.  The evolution equation is the controlled random walk: 
\begin{equation}
X(t+1) = X(t) - U(t) + N(t+1)\, , \quad t \geq 0.
\label{e:CRW}
\end{equation}  
Under the assumption that $\bfmN$ is i.i.d.,  and $\bfmU$ is obtained using a state feedback policy $U(t) = \fee(X(t))$,  the controlled model  is a Markov chain of the form \eqref{e:SP_disc}.


In the experiments that follow we focus exclusively on the average-cost setting, with $c(x,u) = x+u^2/2$, and
\begin{equation}
\fee(x) = \min(x,1+\epsy \sqrt{x})\, ,
\label{e:EpsilonPolicy}
\end{equation}
with $ \epsy > 0$.  This   is  similar in form to the optimal average cost policy calculated in \cite{chehuakulunnzhumehmeywie09}.  It is shown in  \cite{chehuakulunnzhumehmeywie09} that the value function is approximated by the function $h^\theta(x) = \theta^\transpose \psi(x)$ for some $\theta\in\Re^2_+$, with $\psi(x) = (x^{3/2} ,x)^\transpose$.  As in the linear example, the gradient  $\grad \psi(x) = (\frac{3}{2} x^{1/2} , 1)^\transpose$ has much slower growth as a function of $x$.  

Implementation of the $\grad$-LSTD algorithm requires attention to the boundary of the state space.
The sensitivity process $\{\Sens(t)\}$ as defined in \eqref{e:SensDef} requires that the state space be open, and that the dynamics are smooth.  Both of these assumptions are violated in this example.  However,  we do have a representation for the right derivative,
\spm{deleted for CDC}
which evolves according to the recursive equation,
\begin{equation}
\clS(t+1)  = \dffA^\transpose(t+1) \clS(t)  = [1- \ddxp \fee\, (X(t))]\clS(t)  
\label{e:SensSS}
\end{equation}


 We begin with the case in which the marginal of $\bfmN$ is exponential.  In this case the right derivatives and ordinary derivatives coincide a.s..

\subsubsection*{Exponential arrivals}

The marginal distribution of $\bfmN$ was taken to be the unit-mean exponential.

When $\bfmX$ evolves on $\Re_+$ with policy $\fee$ as defined above, the form of the $\grad$-LSTD algorithm is unchanged from  the definition given in \Section{s:TD}.  The recursion for $\bfvarphi$ in \eqref{e:dTD1}
is implemented based on \eqref{e:SensSS}:  
\begin{equation}
\begin{aligned}
\dffA(t+1) 
& =  \one\{X(t) > \bar{\epsy}\} \bigl[ 1 - \half \epsy X(t)^{-1/2} \bigr]
\end{aligned}
\label{e:diffASS}
\end{equation}
where $ \bar{\epsy} = \half ( \epsy + \sqrt{\epsy^2 + 4})$.  

\spm{ok for journal:
Hence  the recursion regenerates:  
$\varphi(t+1)  =  \grad \psi(X(t+1))$ whenever $X(t)\le \bar{\epsy}$.
}

 \notes{Do we need all this? Because the probability of $X(t) = 0$ is $0$??  {I don't think that zero is the sole issue.  We imposed the constraint that the state space is $\Re^d$ throughout, so we have to explain how we can relax this assumption.  I have written an essay above}}

\spm{I don't know about this ... 
which implicitly takes care of the fact that $\Sens(t) = 0$ if $X(t-1) = 0$. }

\spm{commented b defn for CDC}

The regenerative LSTD algorithm used in \cite{chehuakulunnzhumehmeywie09} is not directly applicable in this example.   Various forms of the TD($\lambda$) algorithm were tested, but all appeared to have infinite asymptotic variance.   The introduction of a matrix gain resulted in improved performance.  The examples that follow compare the $\grad$-LSTD algorithm with the best results we were able to obtain using other methods.

The matrix gain algorithm will be called TD-$K$($\lambda$).  It is identical to the standard algorithm, except for the introduction of a matrix gain sequence $\{K_t\}$ in the following:
\begin{equation}
	\begin{aligned}
	\theta(t+1) & = \theta(t) + \gamma_{t+1}  K_t z(t)  d(t+1) 
	\\
	d(t+1) & =  \tilc(t) +  \bigl[\psi(X(t+1)) -  \psi(X(t)) \bigr]^\transpose  \theta(t)
	\\
	\barc(t+1) & = \barc(t) +  \gamma_{t+1}\bigl[-\barc(t) + c(X(t+1)\bigr]
	\\
	z(t+1) & = \lambda z(t) + \psi(X(t+1)).
	\label{e:TD0}
	\end{aligned}
\end{equation}
 $\tilc(t) = c(X(t)) - \barc(t)$.
 \spm{for journal let's use $\tilde\psi$}
To optimize a constant gain $K_t\equiv K$ over all $\ell\times\ell$ matrices,  the solution is obtained by considering the associated ODE,  $\dot \vartheta = KA (\vartheta - \theta^*)$.
The choice $K= -A^{-1}$ is known to be optimal.
\spm{ref to be found}
In this example we have,
\notes{$\vartheta$ is not defined right?
spm: what's wrong?  Note that I do have a minus sign}
\vspace{-0.03in}
\[
A =\Expect\bigl[ z(t) (\psi(X(t+1)) -  \psi(X(t)) )^\transpose \bigr]
\vspace{-0.05in}
\]
where the expectation is taken in steady state, with $z(t) = \sum_{k=0}^\infty \lambda^k \psi(X(t-k))$.   This was estimated using
\begin{equation*}
A_{t+1} = A_t +   \gamma_{t+1}\bigl[ -A_t +  z(t) (\psi(X(t+1)) -  \psi(X(t)) )^\transpose \bigr]
\label{e:TD05}
\end{equation*}
\vspace{-0.05in}
and   $K_{t+1} =-A_{t+1}^{-1}$ (i.e.,   Stochastic Newton Raphson).
\spm{SNR ref needed.  Also...The inverse can of course be computed recursively using the Matrix Inverstion Lemma.}

\begin{figure}[h]  
\Ebox{.85}{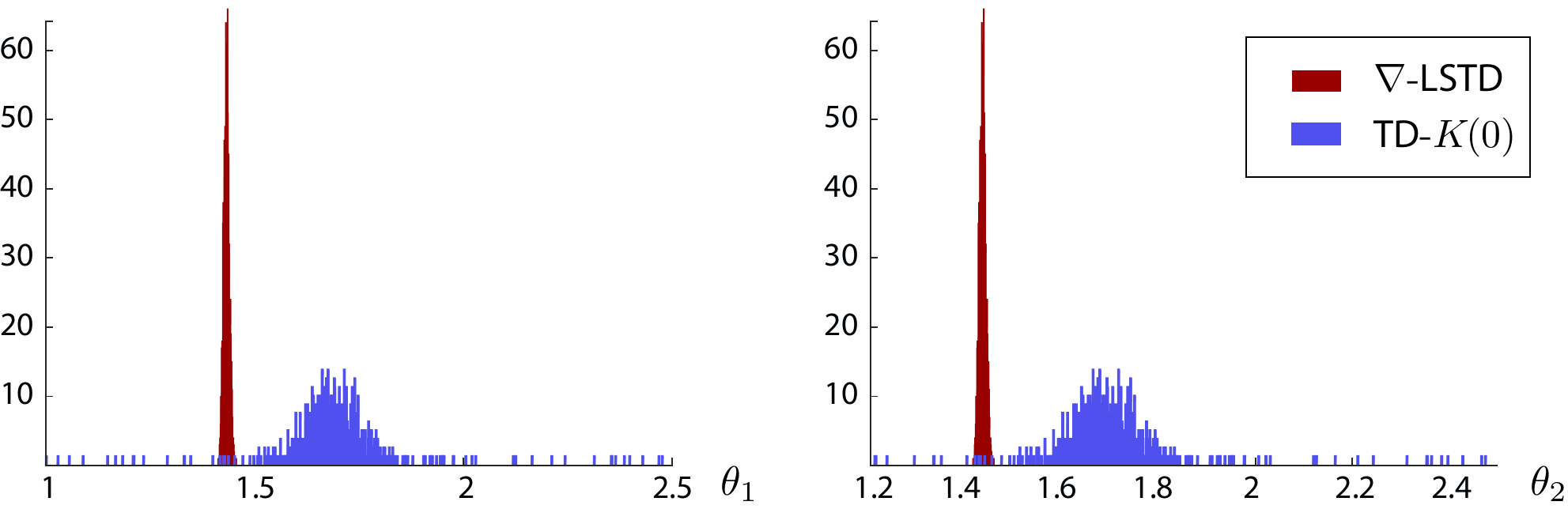}
 	\caption{Histogram of the parameters estimated  using  LSTD and $\grad$-LSTD, for the
	speed scaling model}
	\label{fig:TDvsGTD_Expo1_EpsyPoint5}
\vspace{-.4em}
 \end{figure}
 
\Figure{fig:TDvsGTD_Expo1_EpsyPoint5} shows the histogram of the two parameters, estimated using both $\grad$-LSTD-learning and TD-$K$($0$)-learning, run for a duration of $T=10^5$ time steps. The stationary policy used is as defined in \eqref{e:EpsilonPolicy} with $\epsy = 0.5$.
 Note that here again, the variance reduction obtained using the $\grad$-LSTD-learning algorithm is remarkable. 

\spm{perhaps explain outliers in journal version:
In fact, it was observed that the TD-$K$($0$) algorithm had extremely large variance. Experiments were also run for $\epsy = 0.2$ and $\epsy = 0.8$, and the variance of the final estimates  were more or less the same.
}

\notes{NEED TO CHECK THE MATRIX EQUATIONS!! $b$ or $b^\transpose$?? All that!!}

\subsubsection*{$\grad$-LSTD and regenerative LSTD}

In \cite{chehuakulunnzhumehmeywie09}, the authors consider a discrete state space, with  $N(t)$ geometrically distributed on an integer lattice $\{0, \DeltaA, 2\DeltaA,\dots\}$.
\notes{removed {scaled} non-negative integers}
 In this case, the  theory developed for the $\grad$-LSTD-algorithm does not fit the model since we have no convenient representation of a sensitivity process. 
Nevertheless, the algorithm can be run by replacing gradients with ratios of  differences.  In particular, in implementing the algorithm we substitute the definition \eqref{e:diffASS} with
$\dffA(t) = 1 -   [ \fee(X(t) + \DeltaA) - \fee(X(t))  ]/ \DeltaA$,
and $\grad c$ was approximated similarly.


\notes{DOUBT!! In the experiment we consider, we DO have an expression for grad c and grad u! But in general, we may not (discrete policy). So what do we say here? That we use the differences formula or the actual derrivative?? That's what I have done now. Let me know it has to be changed.
{Don't worry -- I like it!  There is motivation.  Suppose the state space really was discrete, and we didn't have analytic formula for $\fee$}
}
 
The geometric distribution gives $\Prob (N(t)=n\DeltaA) = (1-p_A)^{n} p_A$;
the values $p_A = 0.04$ and $\DeltaA = 1/24$ were chosen, so that $\Expect[N(t)] = 1$.

\spm{can we skip this?
The basis was $\psi(x) = (x^{3/2} ,x)^\transpose$, as previously.
 \\
 I also skipped the definition of regeneration for CDC}


The sequence of steps followed in the regenerative LSTD-learning algorithm are similar to  \eqref{e:TD2}:
\begin{equation}
\begin{aligned}
\varphi(t) & = \one\{X(t-1) \neq 0\} \varphi(t-1) +  \tilpsi(t)
\\
b(t) & = (1-\gamma_t) b(t-1) + \gamma_t \tilc (t)  \varphi(t)
\\
M(t) & = (1-\gamma_t) M(t-1) + \gamma_t \tilpsi(t) \tilpsi^\transpose(t)
\\
\barc(t) & = (1-\gamma_t) \barc(t-1) +  \gamma_t c(X(t)),
\label{e:RTD}
\end{aligned}
\end{equation}
where $\tilc(t) = c(X(t)) - \barc(t)$, 
$\tilpsi(t) \eqdef \psi(X(t)) - \eta_\psi(t)$,  
and
\[
\eta_\psi(t) = (1-\gamma_t) \eta_\psi(t-1) + \gamma_t \psi(X(t)) \, .
\]
The parameter at time $t$ is obtained as $\theta(t) = M^{-1}(t) b(t)$.
\spm{deleted for CDC: Note that the eligibility vector $\varphi(t)$ is reset to $0$ every time the queue empties.}

We replace $\psi$ with $\tilpsi$ in \eqref{e:RTD} to restrict the growth rate of the eligibility vector $\varphi(t)$, which in turn reduces the variance of the estimates $\theta(t)$. This is justified because $h^\theta_a = \theta^\transpose \tilpsi$ differs from $h^\theta_b = \theta^\transpose \psi$ by only a constant term.

\spm{removed to save space: , and the relative value function is unique only up to additive constants.
}

\spm{to avoid all this repetition you need to put all assumptions at the top:
 The stationary policy was once again as defined in \eqref{e:EpsilonPolicy}, with $\epsy = 0.5$. 
 }

 
 \begin{figure}[h]
\Ebox{.85}{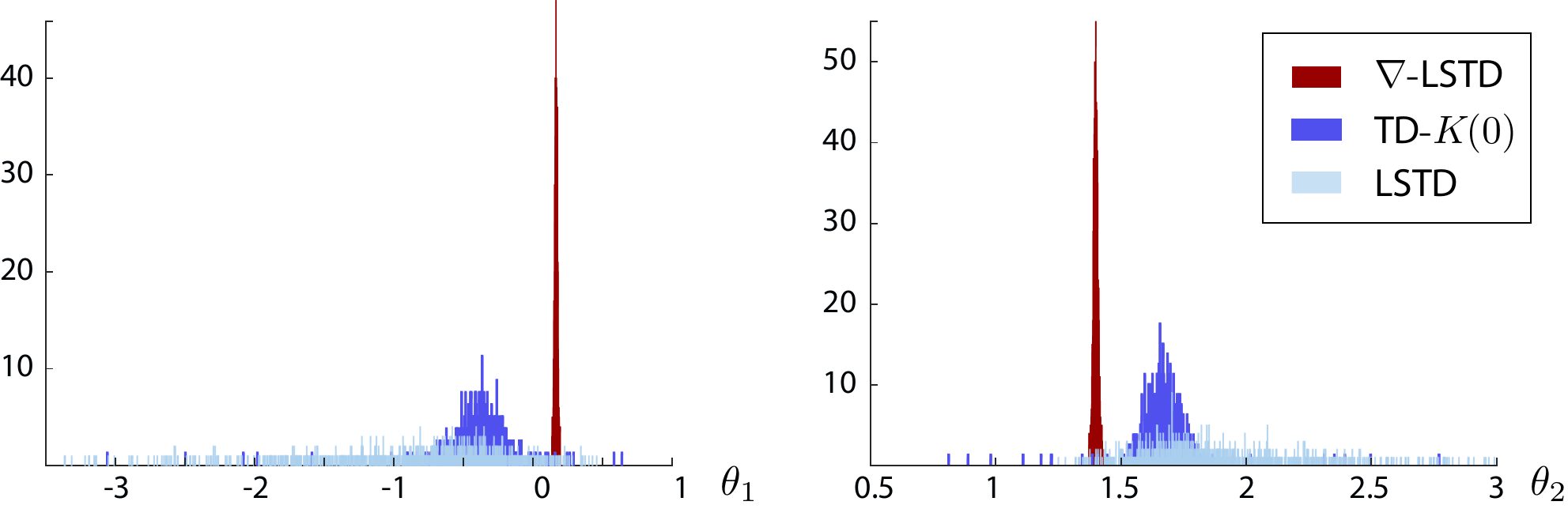}
		\caption{Histogram of the parameters estimated  using the three algorithms.
		}
	\label{fig:VarAnaEpsyPt5T10e5}
 \end{figure}

 For comparison purposes, we also implemented the TD-$K$($0$) algorithm, defined in \eqref{e:TD0}. 
 \Figure{fig:VarAnaEpsyPt5T10e5} shows the histogram of $\theta(T)$ obtained using $\grad$-LSTD-learning, regenerative LSTD-learning, and TD-$K$($0$) learning algorithms, with $T=10^5$. 
 The variance using $\grad$-LSTD-learning algorithm is extremely small compared to the other two.  Also, TD-$K$($0$)   had the largest outliers. 
 \spm{ok?}

It is also noticeable in \Figure{fig:VarAnaEpsyPt5T10e5} that there is a difference in the values to which the   algorithms have converged. To investigate the quality of the estimates through a different lens,  the  \emph{Bellman error} was computed for each algorithm: 
\vspace{-0.05in}
\[
	\clE_B(x) = [P - I]h(x) + \tilc(x),
	\vspace{-0.05in}
\]
where $P$ of course depends on the policy $\fee$, and $h = {\bar{\theta}}^\transpose \psi$, where ${\bar{\theta}}$ is the mean of the $10^3$ parameter estimates obtained for each of the different algorithms.


 \spm{Revised, since this statement suggests that the answer is not unique:    is because there are two degrees of freedom for both these algorithms (two parameters to be estimated), and it is the combination of the two that matters. }

\Figure{TilPsiAndTilC_EpsyPoint5} shows plots of the   Bellman error observed, for each of the three algorithms, for typical values of $\theta(T)$, with 
  $T=10^3$, $10^4$ and $10^5$.  In each case the stationary policy \eqref{e:EpsilonPolicy}  was used, with $\epsy = 0.5$.

The limit of the Bellman error is the same in each experiment.  In the case of the $\grad$-LSTD algorithm,  the Bellman error is unchanged for $T\ge 10^3$.   Achieving similar performance using either of the other algorithms required about $10^5$ samples.   

\begin{figure}[h]
\Ebox{.85}{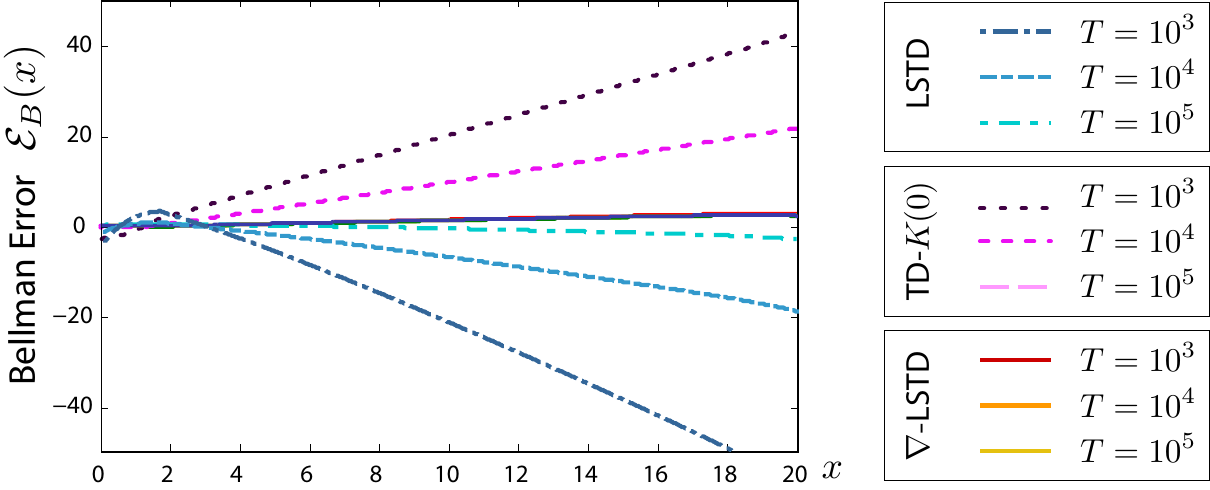}
	\caption{Bellman error corresponding to the estimates of $h$ based on the three algorithms.   
	The convergence time for $\grad$-LSTD algorithm is two orders of magnitude faster than the two other algorithms.    }
		\label{TilPsiAndTilC_EpsyPoint5}
\end{figure}


\section{Conclusions}
\label{sec:conclusions}

The new gradient based TD-learning algorithms for value function approximation introduced here show  remarkable variance reduction in the examples considered.  This is explained by the reduction in magnitude of the functions used as inputs to the algorithm, and also from the additional ``discounting'' that is inherent in the new algorithms.

The most interesting open problem is why the algorithm is so effective even in a discrete state-space setting in which there is no theory to justify its application.

 \notes{Some comments about the advantages of this new approach: Something on sample path representation, variance reduction, and stuff. 
Something about estimating the higher order derivatives of the solutions to Poisson's equations? We give a formula and tell that this can be used?
Also, how we can integrate and get back the solution to Poisson's equation, if required.
\\
{None of the following has been theoritically verified; Also needs to be put better}
 TD-learning for estimating higher order gradients of the solution to Poisson's equation can be obtained, which reduces the variance and increases the convergence rate. Of course, higher the order of the derivatives, the error in estimating the actual original solution could be higher, but higher the order of the original Poisson's equation, less affected is the TD-learning algorithm for estimating higher order derivatives.
 }
 
  \notes{In the formula for higher order derivatives, the lower derivatives appears! Need to think about this! Maybe there is no use of estimating higher order derivatives when we need to estimate lower order derivatives and use it to obtain the higher order derivatives!}
 
%
\bibliographystyle{IEEEtran}
\bibliography{strings,markov,markovExtras,q}


\end{document}